\documentclass[10pt,9pt,conference]{IEEEtran}
\usepackage[T1]{fontenc}
\usepackage[latin9]{inputenc}
\usepackage{geometry}
\usepackage{verbatim}
\usepackage{float}
\usepackage{amsmath}
\usepackage{amsthm}
\usepackage{amssymb}
\usepackage{graphicx}

\makeatletter
 \theoremstyle{definition}
 \newtheorem*{defn*}{\protect\definitionname}
  \theoremstyle{definition}
  \newtheorem{defn}{\protect\definitionname}
  \theoremstyle{plain}
  \newtheorem{prop}{\protect\propositionname}
  \theoremstyle{plain}
  \newtheorem{lem}{\protect\lemmaname}
  \theoremstyle{remark}
  \newtheorem{rem}{\protect\remarkname}
\theoremstyle{plain}
\newtheorem{thm}{\protect\theoremname}

\usepackage{pbox}
\usepackage{multirow}
\usepackage{cite}
\usepackage{float}
\usepackage{endnotes}
\usepackage{color}

\title{\vspace{12pt}Coalitional game with opinion exchange}
\author{Bomin Jiang, Mardavij Roozbehani and Munther A. Dahleh
\thanks{B. Jiang, M. Roozbehani and M.A. Dahleh are with Laboratory for Information and Decision Systems, Institute for Data, Systems, and Society, Massachusetts Institute of Technology. 
M.A. Dahleh is also with the Department of Electrical Engineering and Computer Science. 
{\tt\small \{bominj, mardavij, dahleh\}@mit.edu}}%
}

\makeatother

  \providecommand{\definitionname}{Definition}
  \providecommand{\lemmaname}{Lemma}
  \providecommand{\propositionname}{Proposition}
  \providecommand{\remarkname}{Remark}
\providecommand{\theoremname}{Theorem}

\begin{document}
\maketitle
\begin{abstract}
In coalitional games, traditional coalitional game theory does not
apply if different participants hold different opinions about the
payoff function that corresponds to each subset of the coalition.
In this paper, we propose a framework in which players can exchange
opinions about their views of payoff functions and then decide the
distribution of the value of the grand coalition. When all players
are truth-telling, the problem of opinion consensus is decoupled from
the coalitional game, but interesting dynamics will arise when players
are strategic in the consensus phase. Assuming that all players are
rational, the model implies that, if influential players are risk-averse,
an efficient fusion of the distributed data is achieved at pure strategy
Nash equilibrium, meaning that the average opinion will not drift.
Also, without the assumption that all players are rational, each player
can use an algorithmic R-learning process, which gives the same result
as the pure strategy Nash equilibrium with rational players.
\end{abstract}

\section{Introduction}

In recent years, the application of game theory in multi-agent systems
has been receiving increasing attention. Those applications include
task allocation \cite{shehory1998methods}, smart grids \cite{saad2012game},
transportation networks \cite{saad2011coalition}, sensor placement
\cite{jiang2015optimal}, and so on. Although the theory of coalitional
games has existed for a few decades, theories for the case of unrealized
payoff functions (of subsets of players) are quite limited. In most
of the literature that considers coalitional game theory, an oversimplified
assumption is used, i.e., that all players agree on a common sub-coalition
payoff function.

Recently, researchers have started to look at this case in a variety
of ways, e.g., using the model of Bayesian games, bargaining games,
or repeated playing dynamic games. One paper \cite{ieong2008bayesian}
derived a model that generalizes coalitional games to a Bayesian framework
using types. Furthermore, a Bayesian core contract is defined as the
set of contracts of payoff distributions (Note that a contract is
an agreement among players about the payoff functions of players once
the type is realized) that are non-blocking under the expected value
of payoffs of players, whether Ex ante, Ex interim, or Ex post. Note
that non-blocking means one player is better off staying in the grand
coalition, so this player would not block the formation of such a
coalition. Similarly, another paper \cite{hooper2010coalition} defined
the concept of Bayesian core under uncertainty and gave a bargaining
algorithm that converges to the Bayesian core, assuming that it exists.
However, there are two practical issues with the setting. First, the
theory says nothing when such a core does not exist; second, even
if it does exist, people\textquoteright s individual observations,
which are private, are not used constructively because they do not
exchange information on private payoff functions. By exchanging information,
everyone can obtain a better estimate of the ground truth of the payoff
function. In addition, players will not follow the algorithm in the
literature when they are strategic and want the algorithm to converge
to some value in the core that favors them. Finally, a fair distribution,
such as the Shapley value in the classical coalitional game model,
is not well defined because a commonly-accepted payoff function may
not exist. Another paper, \cite{chalkiadakis2012sequentially}, used
a repeated playing model and assumed that players learn the actual
state of the world as the game goes on, but, in practice, states may
never converge if the game that is being played is changing rapidly
over time or, even worse, if the game is only played once. 

In reality, the realization of such a subcoalition payoff function
may involve opinion consensus, i.e., people\textquoteright s views
of each other are affected by each other, and consensus eventually
reveals the truth. However, to date, there has been virtually no work
on the interplay between coalitional games and opinion consensus theory.
This paper takes an initial step in this direction and shows that
this model gives rise to several interesting implications parallel
to many social phenomena. As noted before, in this model, players
obtain a better estimation of the ground truth of the payoff function
by exchanging information; a fair value distribution (i.e., the Shapley
value) is also well defined given some conditions for efficient opinion
exchange that are stated in the paper.

The proposed framework of the coalitional game with information exchange
results in three interesting phenomena that relate to psychology and
sociology. First, at the equilibrium of this game, each participant
should be a little overconfident by exaggerating their own contribution
in the coalition. Second, in a rational player setting, if the members\textquoteright{}
influences in a network are proportional to their risk-averse levels,
the opinion exchange process is efficient, i.e., it is beneficial
to an organization as a whole if more responsible people are taking
more important positions. Gradual opinion exchange, instead of an
instant opinion fusion, is necessary when players are not fully rational.%
{} 

The rest of this paper is organized as follows: Section II discusses
existing models on coalitional games and opinion consensus. In this
section, Proposition 1 shows that with large number of players, the
Bayesian core as defined in \cite{chalkiadakis2007coalition} is likely
to be empty, and Proposition 2 gives conditions for which linear opinion
consensus could give non-empty Bayesian core. However, the linear
opinion consensus models do not consider self-interested players,
so a modified model with self-interested players is also discussed.
In addition, Section III discusses system dynamics with self-interested
players in coalitional games with opinion exchange. In this section,
Theorem 1 gives conditions for which a coalitional game with opinion
exchange is efficient. Furthermore, conditions for which a Bayesian
core is non-empty are given in Proposition 4 with a consensus assumption,
and Theorem 2 without the consensus assumption. Additionally, Section
IV shows that an R-learning algorithm can provide a player the best
strategy when other players are not rational, i.e., when other players'
behaviors have to be learned. Additionally, potential real world applications
in both equity distribution and legislation lobbying are given in
Section V. Finally, Section VI gives concluding remarks and discusses
further work. 

\section{Coalitional games and opinion consensus models}

In a classical coalitional game, it is assumed that the sub-coalition
payoff function is common knowledge for all players. In this paper,
this oversimplified assumption is removed, and different sub-coalition
payoff functions are allowed. Those different sub-coalition payoff
functions represent different evaluations of other players\textquoteright{}
abilities, i.e., they are private opinions. As indicated in much of
the social science literature, people\textquoteright s opinions can
affect each other substantially \cite{chapman2000power}. Thus, such
opinion exchange requires a new coalitional game model. Our paper,
in particular, uses the linear opinion consensus model \cite{albi2014boltzmann}
as a tool to investigate opinion exchange in coalitional games. Informally,
players first carry out opinion consensus, and then they play the
classical coalitional game to decide the fair payoff distribution;
a rigorous mathematical model of this process is given later in this
paper. However, the coalitional game with opinion exchange is more
than a coalitional game after opinion consensus; during the opinion
consensus process, each participant is incentivized by her or is final
payoff in the coalitional game and may tell lies. That interaction
generates a coupling between the coalitional game and the opinion
consensus. 

\subsection{Notations and definitions}

This subsection reviews notations and definitions used in classical
coalitional games and opinion consensus models. 
\begin{defn*}
{[}Supermodularity{]} Let $N=\{1,2,\cdots,n\}$ be a set of consecutive
integers. Suppose $f(\cdot)$: $2^{N}\rightarrow\mathbb{R}$ is a
set function. The set function $v(C)$ is supermodular iff one of
the following equivalent conditions holds

1. $\forall X\subseteq Y\subseteq N$ and $x\in N\backslash Y$, there
holds $f(X\cup\{x\})-f(X)\leq f(Y\cup\{x\})-f(Y)$, or

2. $\forall X,Y\subseteq N$, there holds $f(X\cup Y)+f(X\cap Y)\geq f(X)+f(Y)$. 

The set function is said to be strictly supermodular if the inequalities
in the above two equations are strict. 
\end{defn*}
\begin{defn*}
{[}Stochastic Matrix{]} A matrix $W=[w_{ij}]$ is called a stochastic
matrix iff 

1. $\forall i,j$, $w_{ij}\geq0$, and

2. $\forall i$, there holds $\sum_{j}w_{ij}=1$
\end{defn*}

\subsection{Coalitional game}

The idea of linear consensus has been used extensively in both engineering
systems \cite{Jiang2017Coalition} and social networks \cite{albi2014boltzmann}.
Let $N=\{1,2,\cdots,n\}$ be a set of $n$ players. In the classical
coalitional game setting, a subset $C\subseteq N$ is called a sub-coalition.
A set function $v(C):2^{N}\rightarrow\mathbb{R}$ of the subcoalition
gives the payoff if sub-coalition $C$ is formed. Note that the cardinality
of $\{C|C\subseteq N\}$ is finite, so $v(C)$ can be represented
as a vector, $v$. 

In a coalitional game, we consider two major questions: 

1. Is there a payoff allocation such that everyone is better-off in
the grand coalition? (This problem is solved by the notion of core),
and

2. Is there a payoff allocation which is fair to everyone? (This problem
is solved by the notion of Shapley value). 

The core of a coalitional game is the set of payoff allocation, $g_{i}$,
$i\in N$:

\[
\{g_{i}|\forall C\subseteq N,\sum_{i\in C}g_{i}\geq v(C),\mbox{ and }\sum_{i\in N}g_{i}=v(N)\}
\]

If the core of a coalitional game is not empty, then the coalitional
game has a stable solution, such that everyone is better off staying
in the grand coalition. Furthermore, the core is non-empty as long
as the payoff function is supermodular. 

The Shapley value is a payoff allocation, $g_{i}$, derived from three
fairness principles, i.e., symmetry, linearity, and null player. This
value is given by
\[
g_{i}=d_{i}(v)=\sum_{C\subseteq N\setminus\{i\}}\frac{|C|!(n-|C|-1)!}{n!}(v(C\cup\{i\})-v(C))
\]
The Shapley value defines a fair distribution of the total payoff
$v(N)$. 

Neither the stable core nor the Shapley value is well defined without
a subcoalition payoff function, $v(C)$, which is commonly accepted
among all players. However, the notion of a stable core is generalized
in \cite{chalkiadakis2004bayesian} and \cite{ieong2008bayesian}
as a Bayesian core, where private sub-coalition payoff functions are
assumed. Note that the definitions of the Bayesian core are different
in the above two papers, and the definition in our paper is similar
to that in \cite{chalkiadakis2004bayesian}. 
\begin{defn}
The Bayesian core is defined as the set of value distributions, $d_{i}$,
such that every player is better off staying in the grand coalition.
Mathematically, \textquotedblleft better off\textquotedblright , or
rationality, is defined as
\begin{equation}
\forall i,\forall C\subsetneqq N,\sum_{j\in C}g_{j}\geq v_{i}(C)\label{eq:better off}
\end{equation}
where $v_{i}$ is private information of player $i$, charactering
his or her unique opinion of the game. Furthermore, there holds the
budget constraint
\begin{equation}
\forall i,\sum_{j\in N}g_{j}\leq v_{i}(N)\label{eq:budget}
\end{equation}
Now, a value distribution, $d_{i}$, is in the Bayesian core iff both
\eqref{eq:better off} and \eqref{eq:budget} hold. 

The problem with the setting of the Bayesian coalitional game is that,
in many cases, the Bayesian core is empty even though the core is
not empty for each player $i$. That is particularly true if the number
of players, $n$, is large, as illustrated by Lemma \ref{lem: no core}. 
\end{defn}
\begin{prop}
\label{lem: no core} Suppose that there is a strictly supermodular
ground truth payoff function $v(C):2^{\Omega}\rightarrow\mathbb{R}$
(Note one can represent the function as an $m$-vector). Further suppose
each player's opinion $v_{i}\sim N(v,\Sigma_{i})$ is a sample from
the ground truth payoff function, where $N(v,\Sigma_{i})$ represents
a truncated normal distribution with support $v_{i}\in\left\{ V|V\text{ vectorize }v_{i}(C)\text{ and }v_{i}(C)\text{ is supermodular}\right\} $,
and $\Sigma_{i}$ is a diagonal matrix with diagonal entries $\sigma_{i}^{2}>0$.
As the number of players increases, i.e. $n\rightarrow\infty$, the
Bayesian core defined by \eqref{eq:better off} and \eqref{eq:budget}
is empty with probability 1. 
\end{prop}
\begin{proof}
Proof by contradiction. Let $S_{1}$ and $S_{2}$ be partitions of
$N$. i.e., $S_{1}\cap S_{1}=\emptyset$ and $S_{1}\cup S_{2}=N$.
Because each player takes a sample from a Gaussian distribution, 
\[
\lim_{n\rightarrow\infty}\mathbb{P}\{\exists i,j,k\text{, s.t. }v_{i}(S_{1})>\frac{1}{2}v_{k}(N)\text{ and }v_{j}(S_{2})>\frac{1}{2}v_{k}(N)\}\rightarrow1
\]
 The ``better off'' condition defined by \eqref{eq:better off}
gives
\[
v_{k}(N)<v_{i}(S_{1})+v_{j}(S_{2})\leq\sum_{p\in S_{1}}g_{p}+\sum_{p\in S_{2}}g_{p}=\sum_{p\in N}g_{p}
\]

Now the inequality $\sum_{p\in N}g_{p}>v_{k}(N)$ contradicts the
budget constraint defined by \eqref{eq:budget}. 
\end{proof}
In the above proposition, even if each sampled value function has
non-empty core, the game itself has empty Bayesian core. 

\subsection{Opinion consensus}

\label{subsec:opinion consensus}

The idea of linear consensus has been used extensively in both engineering
systems \cite{jiang2017simultaneous} and social networks \cite{albi2014boltzmann}.
Suppose a graph $\mathcal{G}=\{N,E\}$ characterizes the opinion influence
among a set of players $N$. There is an edge $e_{ij}\in E$ with
weight $w_{ij}\in(0,1)$ if player $i$ has a influence on player
$j$'s opinion. If there is not an edge between $i$ and $j$, we
set $w_{ij}=0$. Furthermore, $w_{ii}=1-\sum_{j\neq i}w_{ij}\geq0$.
At each time instance $t_{k}$, player $i$ hold an opinion of the
function $v_{i}(C)[k]$. In fact, because the cardinality of $\{C|C\subseteq N\}$
is finite, one can consider $v_{i}(C)[k]$ as a vector $v_{i}[k]$.

In the classical opinion consensus literature, all players are truth-telling.
When all players are truth- telling, the problem of opinion consensus
is decoupled from the coalitional game. Players just update their
opinions according to the linear opinion dynamics defined by

\[
v_{i}[k]=\sum_{j}w_{ij}v_{j}[k-1]
\]
In the above system, opinion consensus can be achieved, i.e., $\forall i$,
the limit $\lim_{k\rightarrow\infty}v_{i}[k]$ exists and $\forall i,j,\lim_{k\rightarrow\infty}v_{j}[k]=\lim_{k\rightarrow\infty}v_{j}[k]$,
iff the stochastic matrix $W=[w_{ij}]$ has one eigenvalue of 1 and
all other eigenvalues are strictly in the unit disk. If the opinion
consensus can be achieved, one can define a consensused payoff function
$v=\lim_{k\rightarrow\infty}v_{i}[k]$. After the opinion consensus,
players can play the coalitional game and a grand coalition exists
iff the stable core of $v$ is non-empty. A set of sufficient conditions
for which the stable core of $v$ is non-empty is given by Proposition
\ref{lem: truth telling core nonempty} below. 
\begin{prop}
\label{lem: truth telling core nonempty} There exists a stable coalition
under the consensused payoff function $v$ if all of the prior payoff
functions $v_{i}[1]$ are supermodular set functions. 
\end{prop}
\begin{proof}
First, note that if $v_{i}[k-1]$ is supermodular, then $v_{i}[k]=\sum_{j}w_{ij}v_{j}[k-1]$
is also supermodular. Since $v_{i}[1]$ is supermodular, by induction,
it follows that $\forall k$, $v_{i}[k]$ is supermodular. 

Because the set of supermodular set functions is closed, the consensused
payoff function $v=\lim_{k\rightarrow\infty}v_{i}[k]$ is supermodular.
Since the stable core is non-empty as long as the payoff function
is supermodular, a stable coalition exists under the consensused payoff
function. 
\end{proof}
In reality, however, because players are incentivized by the payoff
distributions in the coalitional game, they do not necessarily tell
the truth during the opinion consensus process. To incorporate this
strategic aspect of opinion consensus, a better model is required
\cite{tanaka2017faithful}. Assume that, at each time instance, player
$i$ reveals an opinion, $x_{i}[k]$, that may or may not be equal
to $v_{i}[k]$. Define $\theta\in(0,1)$ as a trust parameter. Now
each player updates his or her opinion according to the linear opinion
dynamics defined by

\begin{equation}
v_{i}[k]=\theta\sum_{j}w_{ij}x_{j}[k-1]+(1-\theta)v_{i}[k-1].\label{eq: lie dynamic}
\end{equation}
The above opinion dynamics will be discussed in the next section. 

\section{System dynamics with strategic players \label{sec:System-dynamics-with}}

As pointed out in Section \ref{subsec:opinion consensus}, the system
dynamics are trivial when all players are truth-telling because the
opinion consensus and coalitional game are decoupled. This chapter
discusses the opinion dynamics when players may tell lies to get themselves
better payoff distributions. We refer to such players as strategic
players.

\subsection{Enforcing effective information exchange}

In the rational player setting, if telling a lie has no cost, the
game becomes a cheap-talk game \cite{farrell1987cheap}. Players will
not trust any information, and there is no efficient information exchange.
Similar problems exist under the cognitive hierarchical model \cite{camerer2004cognitive},
where the opinions will not reach consensus because the second level
players again form a cheap-talk game, even though the first level
players may tell the truth. We would like to investigate this type
of bounded rationality models in future work.

In our coalitional game setting, each player's private knowledge
of the sub-coalition payoff function can be viewed as a sample of
the ground truth. People enter the opinion consensus to acquire information
on other samples, and hence acquire a better understanding of the
ground truth. However, revealing false information to others will
introduce bias. Hence, it is useful to introduce a disutility when
false information is revealed so that players become risk-averse,
and effective information exchange is established.

\subsection{Rational and risk-averse players\label{subsec:Fully-rational-players:}}

Consider a ground truth payoff function $v(\cdot)$. Suppose it is
normalized, i.e. $v(\emptyset)=0$ and $v(N)=1$. When we write it
in its vector form, each entries in $v$ are $v(C),\emptyset\subsetneq C\subsetneq N$
(hence it is an $m$-vector, $m=2^{N}-2$). Note $v(C),\emptyset\subsetneq C\subsetneq N$
is unknown to players. Further suppose that each player's private
initial opinion at time instance $k=0$ is an i.i.d. sample $v_{i}[0]\sim N(v,\Sigma_{i})$,
where $\Sigma_{i}=\sigma_{i}^{2}I$ and $I$ is the identity matrix.
Let the influence among players denoted by $W=\left[w_{ij}\right]$.
In addition, define weight of opinions as $t_{i}$ such that $\lim_{k\rightarrow\infty}W^{k}=1_{n}\left[\begin{array}{cccc}
t_{1} & t_{2} & \cdots & t_{n}\end{array}\right]$ ($1_{n}$ denotes $n$ dimensional column-1-vector), then $\hat{v}[0]=\mathbb{A}[v_{i}[0]]=\sum_{i}t_{i}v_{i}[0]$
is an unbiased estimator of $v$. If it happens that $t_{i}=\frac{1}{\sigma_{i}^{2}}/\sum_{i}\frac{1}{\sigma_{i}^{2}}$,
then $\hat{v}[0]$ is the ML estimator. Note that $\mathbb{A}[\cdot]$
denotes weighted average. Note $[w_{ij}]$ and $t_{i}$ are also common
knowledge; the only private information to player $i$ is its opinion
$v_{i}[k]$.

Suppose the payoff $v(N)$ is allocated according to Shapley value
$d_{i}(\hat{v}[K])$ of the average opinion $\hat{v}[K]$ at step
$K$. Because the Shapley ratio defines a linear function, we can
also refer this final payment to player $i$ as $\frac{1}{n}+d_{i}^{T}\hat{v}[K]$,
where $d_{i}^{T}$ is a $m$-vector. Note the property of Shapley
value gives $\sum_{i}d_{i}^{T}=0$. If every player is truth telling,
then the system reaches consensus, i.e. $\forall i,\lim_{k\rightarrow\infty}v_{i}[k]=\hat{v}$,
and the final payoff function is an unbiased estimator.

Now, assume that players can tell lies. Each player may introduce
some fraud at step $k$: $u_{j}[k]=x_{j}[k]-v_{j}[k]$, but, at the
same time, these fraudulent statements undermine trust in the system,
and, hence, they introduce disutility $\boldsymbol{1}^{T}\text{var}[u]$,
where $\text{var}[u]=\sum_{i}t_{i}\left(u_{i}[k]\right)^{2}-\left(\sum_{i}t_{i}u_{i}[k]\right)^{2}$
(Note that this disutility metric is a scalar, and it can be interpreted
as the 1-norm of the variance of $u_{i}[k]$). After $K$ steps of
playing, the overall disutility due to fraud is given by $\boldsymbol{1}^{T}\sum_{k=1}^{K}\text{var}\left[u[k]\right]$.
Each player makes a trade-off between $\boldsymbol{1}^{T}\sum_{k=1}^{K}\text{var}\left[u[k]\right]$
and $d_{i}(\mathbb{A}[v[K]])$ by solving the minimization problem

\begin{equation}
\arg\min_{u_{i}[k],k=1,2,\cdots}p_{i}\cdot\boldsymbol{1}^{T}\sum_{k=1}^{K}\text{var}\left[u[k]\right]-d_{i}^{T}\mathbb{A}[v[K]]\label{eq: objective}
\end{equation}
where $p_{i}$ is the risk-averse factor of player $i$. 
\begin{lem}
\label{lem:v to u}In the coalitional game with opinion exchange (which
follows the system dynamics \eqref{eq: lie dynamic}), there holds
\[
\mathbb{A}[v[K]]=\theta\mathbb{A}[u[K-1]]+\mathbb{A}[v[K-1]]=\sum_{k}\theta\mathbb{A}[u[k]]+\mathbb{A}[v[0]]
\]
\end{lem}
\begin{proof}
Let $T=\left[t_{i}\right]$, $W=\left[w_{ij}\right]$, $V[k]=\left[v_{i}[k]\right]$,
$X[k]=\left[x_{i}[k]\right]$ and $U[k]=\left[u_{i}[k]\right]$. By
definition 
\[
\mathbb{A}[v[k+1]]=T^{\top}V[k+1]
\]
Substitute \eqref{eq: lie dynamic} into the above equation. We obtain
\[
\mathbb{A}[v[k+1]]=T^{\top}\left(\theta WX[k]+(1-\theta)V[k]\right)
\]
By the definition of $u_{i}$ in the last subsection, it holds that
\[
\begin{split}\mathbb{A}[v[k+1]] & =T^{\top}\left(\theta W(V[k]+U[k])+(1-\theta)V[k]\right)\\
 & =T^{\top}\theta WV[k]+T^{\top}\theta WU[k]+(1-\theta)T^{\top}V[k]
\end{split}
\]
According to the definition of $t_{i}$, the influence among players
satisfies $\lim_{k\rightarrow\infty}W^{k}=1_{n}\left[\begin{array}{cccc}
t_{1} & t_{2} & \cdots & t_{n}\end{array}\right]$, i.e., $T^{\top}W=T^{\top}$. Therefore 
\[
\begin{split}\mathbb{A}[v[k+1]] & =T^{\top}\theta V[k]+T^{\top}\theta U[k]+(1-\theta)T^{\top}V[k]\\
 & =T^{\top}V[k]+\theta T^{\top}U[k]\\
 & =\sum_{k}\theta\mathbb{A}[u[k]]+\mathbb{A}[v[0]]
\end{split}
\]
\end{proof}
From Lemma \ref{lem:v to u}, we obtain 
\[
\begin{split} & \arg\min_{u_{i}[k],k=1,2,\cdots}p_{i}\cdot\boldsymbol{1}^{T}\sum_{k=1}^{K}\text{var}\left[u[k]\right]-d_{i}^{T}\mathbb{A}[v[K]]\\
= & \arg\min_{u_{i}[k],k=1,2,\cdots}\sum_{k=1}^{K}\left(p_{i}\cdot\boldsymbol{1}^{T}\text{var}\left[u[k]\right]-\theta d_{i}^{T}\mathbb{A}[u[k]]\right)
\end{split}
\]

That says that the optimal strategy is indeed a myopic strategy. Therefore,
one can seek to find the optimal strategy step-by-step. In step $k$,
it holds that

\[
\begin{split} & \arg\min_{u_{i}[k]}p_{i}\cdot\boldsymbol{1}^{T}\text{var}\left[u[k]\right]-\theta d_{i}^{T}\mathbb{A}[u[k]]\\
= & \arg\min_{u_{i}[k]}p_{i}\cdot\boldsymbol{1}^{T}\left[\sum_{j}t_{j}\left(u_{j}[k]\right)^{2}-\left(\sum_{j}t_{j}u_{j}[k]\right)^{2}\right]\\
 & -\theta d_{i}^{T}\sum_{j}t_{j}\cdot u_{j}[k]
\end{split}
\]

Set the first derivative with respect to $u_{i}$ to zero

\begin{equation}
2p_{i}\left(t_{i}u_{i}[k]-t_{i}\left(\sum_{j}t_{j}u_{j}[k]\right)\right)=t_{i}d_{i}\theta\label{eq:first order condition}
\end{equation}

The above equation defines the best strategy of player $i$ given
the actions of the other players. The linear equations above can be
used to solve for pure strategy Nash equilibrium. Note that the coefficient
matrix of the above linear equations has the rank of $n-1$, so there
are multiple Nash equilibria.

Suppose for now that the weight $p_{i}$ of disutility is proportional
to player $i$'s influence $t_{i}$ in the network, i.e. $p_{i}\propto t_{i}$.
Further because of the property of the Shapley value, $\sum_{j}d_{j}=0$,
one can obtain a solution of \eqref{eq:first order condition} as

\begin{equation}
u_{i}[k]=\frac{d_{i}\theta}{2p_{i}}\label{eq:optimal strategy}
\end{equation}

The above solution is a pure strategy Nash-equilibrium, and it yields
$\forall i,u_{i}\neq0$, and $\sum_{i}t_{i}u_{i}[k]=0$. In addition,
at the equilibrium, $v_{i}[k]$ will converge, but not achieve consensus.
The larger the value of $p_{i}$ is, the smaller the opinion divergence
is, and the more likely that $d_{i}(\hat{v}[\infty])$ is a stable
coalition for all players. 
\begin{rem}
In practice, the assumption of $p_{i}\propto t_{i}$, i.e., the weight
$p_{i}$ of disutility is proportional to player $i$'s influence
$t_{i}$, implies that more responsible players are placed at more
important positions in a network. 
\end{rem}
\begin{defn}
A coalitional game with information exchange is efficient if there
exists a Nash equilibrium such that the average opinion is constant,
i.e. $\hat{v}[k]$ is constant for all time instances $k$. 
\end{defn}
\begin{thm}
\label{thm:invariant}In the fully rational risk-averse player scenario,
i.e., opinion dynamics follows \eqref{eq: lie dynamic} and strategic
players minimize the objective function \eqref{eq: objective}, the
coalitional game with information exchange is efficient if $p_{i}\propto t_{i}$
over all players. 
\end{thm}
\begin{proof}
Let $T=\left[t_{i}\right]$, $W=\left[w_{ij}\right]$, $V[k]=\left[v_{i}[k]\right]$,
$X[k]=\left[x_{i}[k]\right]$ and $U[k]=\left[u_{i}[k]\right]$. By
Lemma \ref{lem:v to u}, 
\[
\hat{v}[k+1]=T^{\top}V[k]+\theta T^{\top}U[k]
\]
Given $p_{i}\propto t_{i}$, the optimal strategy for each rational
risk-averse player is given by \eqref{eq:optimal strategy}. Because
the solution \eqref{eq:optimal strategy} satisfies $\sum_{i}t_{i}u_{i}[k]=0$,
i.e. $T^{\top}U[k]=0$, we find that 
\[
\hat{v}[k+1]=T^{\top}V[k]=\hat{v}[k]
\]
is invariant over time steps $k$. 
\end{proof}

\subsection{Existence of stable coalition \label{subsec:Existence-of-stable}}

This subsection discusses conditions for non-empty Bayesian core.
Assuming that consensus is achieved, Proposition \ref{lem: core for consensus}
gives a sufficient condition for which a stable coalition exists. 
\begin{prop}
\label{lem: core for consensus} Suppose that consensus is achieved.
There is a stable coalition under the consensused payoff function
if all of the prior payoff functions $v_{i}[0]$ and all the reported
payoff functions $x_{i}[k]$ are supermodular set functions. 
\end{prop}
\begin{proof}
First, we want to show that $\forall k$, $v_{i}[k]$ is supermodular.

Let $k\geq2$ be given. If $v_{i}[k-1]$ is supermodular, then $v_{i}[k]=\theta_{i}\sum_{j}w_{ij}x_{j}[k-1]+(1-\theta_{i})v_{i}[k-1]$,
as a positive weighted average of supermodular set functions, is also
supermodular. In addition, because $v_{i}[1]$ are supermodular, by
induction, we know that $\forall k$, $v_{i}[k]$ is supermodular.

Given that all payoff functions in step $k$ are supermodular, and
because the set of supermodular set functions is a closed set, and
further because $v=\lim_{k\rightarrow\infty}v_{i}[k]$, the consensused
payoff function $v$ is supermodular. Furthermore, because the Shapley
value of a supermodular payoff function is in the stable core, we
reach the conclusion that there is a stable coalition under the consensused
payoff function. 
\end{proof}
In the above proposition, the assumption of achieved consensus may
be too strong in practice. Therefore, in Theorem \ref{thm: non empty},
the assumption of achieved consensus is removed, and a stable coalition
is shown to exist when $p_{o}$ is sufficiently large. 
\begin{thm}
\label{thm: non empty} Assume that $p_{i}\propto t_{i}$ over all
players, and define $p_{o}=p_{i}/t_{i}$. Then $v_{i}[\infty]=\lim_{k\rightarrow\infty}v_{i}[k]$
exists. In addition, if $v_{i}[0]$ is strictly supermodular for any
given set of initial states $v_{i}[0],i\in N$, then $\exists p_{o}>0$
s.t. the Bayesian core is non-empty with subcoalition payoff functions
$v_{i}[\infty]$, that is, the Bayesian core is non-empty after the
opinion consensus process. 
\end{thm}
\begin{proof}
Because $\forall i\in N$, $v_{i}[0]$ is strictly supermodular, $\hat{v}[0]$,
the weighted average of all $v_{i}[0]$, is also strictly supermodular.
Further because $p_{i}\propto t_{i}$ over all players, the average
opinion $\hat{v}[k]$ is invariant over time step $k$, hence $\hat{v}[k]=\hat{v}[0]$
is strictly supermodular.

Moreover, given the optimal strategy solution \eqref{eq:optimal strategy},
one can rewrite the system dynamics of \eqref{eq: lie dynamic} as
\[
v_{i}[k]=\theta\sum_{j}w_{ij}\left(v_{j}[k-1]+\frac{d_{i}\theta}{2p_{o}t_{i}}\right)+(1-\theta)v_{i}[k-1]
\]
Define $V[k]=\left[v_{i}[k]\right]$, $X[k]=\left[x_{i}[k]\right]$,
$U[k]=\left[\frac{d_{i}\theta}{2p_{o}t_{i}}\right]$, $W=\left[w_{ij}\right]$
and $\overline{W}=\theta W+(1-\theta)I$. Because $W$ is a stochastic
matrix, the limits $\lim_{k\rightarrow\infty}W^{k}$ and $\lim_{k\rightarrow\infty}\overline{W}^{k}$
exist and are equal to each other. We define $T=\lim_{k\rightarrow\infty}W^{k}=\lim_{k\rightarrow\infty}\overline{W}^{k}$.
A stochastic matrix has a eigenvalue equals 1 and all other eigenvalues
inside the unit disk, so one can define an eigenvalue decomposition
$\overline{W}=D^{-1}\left[S_{1}+S_{2}\right]D$, where $S_{1}=\left[\begin{array}{cccc}
1 & 0 & \cdots & 0\\
0 & 0 & \cdots & 0\\
\vdots & \vdots & \ddots & \vdots\\
0 & 0 & \cdots & 0
\end{array}\right]$ and $S_{2}$ is a diagonal matrix with the first entry 0 and all
other entries inside the unit disk. The system dynamics is given by
\[
\begin{split}V[k] & =\theta W\left(V[k-1]+U\right)+(1-\theta)V[k-1]\\
 & =\overline{W}V[k-1]+\theta WU\\
 & =\overline{W}^{k}V[0]+\theta W\left(I+\overline{W}+\overline{W}^{2}+\cdots+\overline{W}^{k-1}\right)U
\end{split}
\]
If we further consider the eigenvalue decomposition $\overline{W}=D^{-1}\left[S_{1}+S_{2}\right]D$,
we obtain 
\[
\begin{split}V[k] & =\overline{W}^{k}V[0]+\theta WD^{\top}\left[I+\left(S_{1}+S_{2}\right)+\left(S_{1}+S_{2}\right)^{2}\right.\\
 & \left.+\cdots\left(S_{1}+S_{2}\right)^{k-1}\right]DU\\
 & =\overline{W}^{k}V[0]+\theta WD^{\top}\left[I+S_{2}+S_{2}^{2}+\cdots S_{2}^{k-1}\right]DU\\
 & +(k-1)\theta WTU
\end{split}
\]
Because the solution \eqref{eq:optimal strategy} satisfies $\sum_{i}t_{i}u_{i}=0$,
i.e. $TU=0$, it holds that 
\[
\begin{split}V[k] & =\overline{W}^{k}V[0]+\theta WD^{\top}\left[I+S_{2}+S_{2}^{2}+\cdots S_{2}^{k-1}\right]DU\end{split}
\]
Considering the fact that $S_{2}$ is a diagonal matrix with all entries
in the unit disk, when $k\rightarrow\infty$, the series $I+S_{2}+S_{2}^{2}+\cdots S_{2}^{k-1}$
converges to $(1-S_{2})^{-1}$. Further because $\lim_{k\rightarrow\infty}W^{k}=T$,
we obtain 
\[
\lim_{k\rightarrow\infty}V[k]=TV[0]+\theta WD^{\top}(1-S_{2})^{-1}DU
\]
Recall that we have $U=\left[\frac{d_{i}\theta}{2p_{o}t_{i}}\right]$,
hence, 
\[
\lim_{p_{o}\rightarrow\infty}\lim_{k\rightarrow\infty}V[k]=TV[0]
\]
i.e., 
\[
\lim_{p_{o}\rightarrow\infty}v_{i}[\infty]=\hat{v}[\infty]
\]

Given that $\hat{v}[\infty]$ is strictly supermodular, the Shapley
value $d(\hat{v}[\infty])$ is in the interior of the core of the
coalitional game with subcoalition payoff functions $\hat{v}[\infty]$.
For each player $i$, $\lim_{p_{o}\rightarrow\infty}v_{i}[\infty]=\hat{v}[\infty]$
and the mapping from $v_{i}[\infty]$ to the core is continuous. Hence,
for sufficiently large $p_{o}$, $d(\hat{v}[\infty])$ is also in
the core of the coalitional game with subcoalition payoff function
$v_{i}[\infty]$, concluding the proof. 
\end{proof}

\section{Algorithmic playing\label{sec:Algorithmic-playing}}

In the above analysis, we made the assumption that all of the players
are rational. In the case in which all players are fully rational
and risk-averse, an equilibrium exists and convergence can be achieved.
However, from a single player\textquoteright s perspective, he or
she does not have control over how the other players play. What should
a player do if he or she is fully rational while others are not? In
this section, we show that an R-learning algorithm \cite{sutton1998reinforcement}
can provide such a player the best strategy.

\subsection{R-Learning Formulation}

At each step $k$, the reward of a rational player is given by:

\[
r_{i}[k]=-p_{i}\cdot\boldsymbol{1}^{T}\text{var}\left[u[k]\right]+\theta d_{i}^{T}\mathbb{A}\left[u[k]\right]
\]

Define $s[k]=\mathbb{A}\left[v[k]\right]$ as the state and $u_{i}$
as the action of each player $i$. For convenience, let $i-$ denote
the players other than player $i$. Furthermore, because all of the
state variables and action variables are continuous, a model of environment
(i.e., action pattern of players $i-$) with finite parameters must
be defined prior to the learning process. Thus we define the environment
$e=\left[\begin{array}{c}
\mathbb{A}\left[u_{i-}[k]\right]\\
\text{var}\left[u_{i-}[k]\right]
\end{array}\right]$, and the environment model $p_{e}(e|s,\text{var}\left[u[k-1]\right];\Phi)$
as the probability distribution of $e$ given state $s$, parameterized
by $\Phi$. Note $\Phi$ is the set of finite environment parameters
to be learned. In addition, the environment $e$ is independent of
current action $u_{i}$, but the rewards and next state are functions
of environment $e$ and the action $u_{i}$. 

\begin{figure}[H]

\includegraphics[width=0.9\columnwidth]{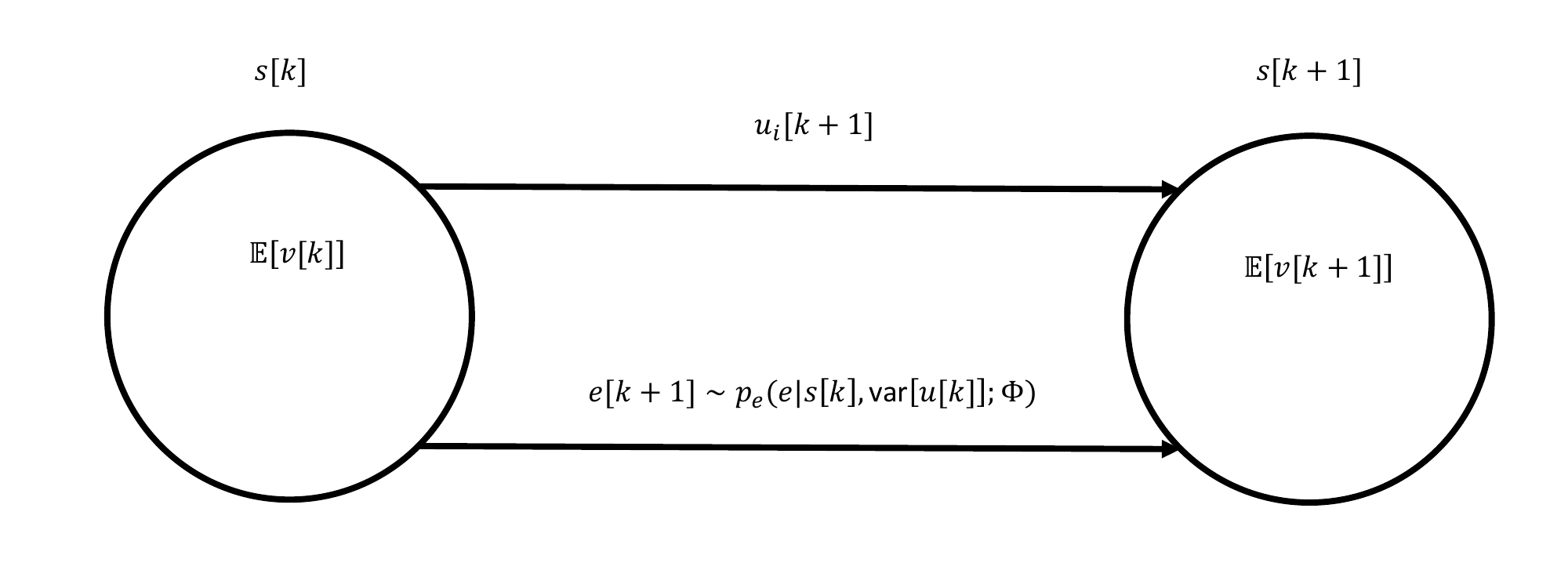}\caption{States transition of R-learning}

\end{figure}

One may find it problematic that the states and the associated rewards
are not observable for player $i$, hence the learning process cannot
proceed unless $\text{var}\left[u[k]\right]$ and $\mathbb{A}\left[v[k]\right]$
are broadcast centrally. Furthermore, $\mathbb{A}\left[v[k]\right]$
cannot be obtained so even a central broadcast would be problematic.
However, the rewards in each step depend only on the decision variable
and environment, but not directly on any state variable; i.e., the
impact of state variables only goes into the system via environment
$e$. As a result, the choice of state variables in the R-learning
process depends only on how $i-$ players are modeled, and it is possible
to choose state variables other than $\mathbb{A}\left[v[k]\right]$,
e.g. $\mathbb{A}\left[x[k]\right]$. 
\begin{rem}
In Section \ref{sec:Algorithmic-playing}, if everyone is rational
and adopts the R-learning algorithm, then \eqref{eq:optimal strategy}
is the optimal strategy. In this case, although Section \ref{sec:Algorithmic-playing}
and Subsection \ref{subsec:Fully-rational-players:} have the same
objective function and the same optimal strategy, some assumptions
are different, i.e., Subsection \ref{subsec:Fully-rational-players:}
assumes that all players know that all players are rational. However,
Section \ref{sec:Algorithmic-playing} does not have this assumption,
but it requires that $\text{var}\left[u[k]\right]$ and $\mathbb{A}\left[v[k]\right]$
can be broadcast centrally. 
\end{rem}
\begin{rem}
The learning process justifies the need for gradual consensus of opinion,
i.e., the participants learn each other\textquoteright s patterns
during the consensus process.
\end{rem}

\subsection{Simulations}

As an illustrative example, first, we look at a two-player coalitional
game with opinion exchange. Suppose player 2's expressed opinion is
quasilinear in its true opinion and depends on the mean opinion, i.e.
$x_{2}[k+1]=v_{2}[k+1]+f\left(\mathbb{A}\left[x[k]\right]\right)+w$
where $w$ is white noise. Further, assume that player 1 is a risk-averse,
rational player as defined in Subsection \ref{subsec:Fully-rational-players:},
and uses an R-learning algorithm to learn the $f(\cdot)$ function
during the opinion consensus process to maximize his or her own utility
in the coalitional game. 

From player 1's perspective, his or her optimal strategy is given
by the solution of \eqref{eq:first order condition}

\[
2p_{1}\left(u_{1}[k]-\left(t_{1}u_{1}[k]+t_{2}u_{2}[k]\right)\right)=d_{1}\theta
\]
\[
u_{1}^{*}[k+1]=\frac{d_{1}\theta}{2p_{1}(1-t_{1})}+\frac{t_{2}\bar{f}\left(\mathbb{A}\left[x[k]\right]\right)}{1-t_{1}}
\]
where $\bar{f}\left(\cdot\right)$ is player 1's estimate of $f\left(\cdot\right)$. 

In the simulation, assume that $v(\left\{ 1,2\right\} )=1$, $v(\left\{ \emptyset\right\} )=0$
, $v_{i}[k]=\left[\begin{array}{c}
v(\left\{ 1\right\} )\\
v(\left\{ 2\right\} )
\end{array}\right]$. Let the initial conditions be $v_{1}[0]=\left[\begin{array}{c}
0.7\\
0.1
\end{array}\right]$ and $v_{2}[0]=\left[\begin{array}{c}
0.3\\
0.5
\end{array}\right]$. Furthermore, let the system parameters be $\theta=0.1$, $W=\left[\begin{array}{cc}
0.3 & 0.7\\
0.4 & 0.6
\end{array}\right]$. Player 1 has the probability of $\gamma=0.5$ of implementing the
optimal strategy given its current estimate $\bar{f}\left(\cdot\right)$
(exploitation), and this player has the probability of $1-\gamma$
of carrying out exploration. Further assume that both players are
rational and risk-averse, but do not know that their opponents are
rational. The coalitional game with information exchange in this case
will be efficient, i.e. $d(\hat{v}[k])$ is invarient over $k$, as
shown in the example in Figure \ref{fig:True-opinion-}. 

\begin{figure}
\begin{centering}
\includegraphics[width=0.42\textwidth]{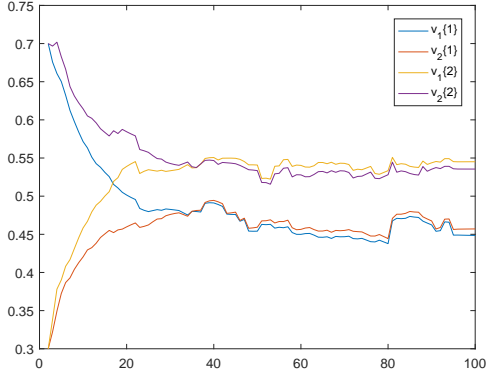} 
\par\end{centering}
\vspace{-10pt}

\caption{True opinion $v[k]$ when $\gamma=0.8$. Here both players are rational
and risk-averse. Both players are doing R-learning to learn the behavior
of their opponents. \label{fig:True-opinion-}}

\vspace{-10pt}
\end{figure}

\section{Real world applications}

Deciding equity distribution is a critical step in forming a startup
company \cite{rose2016research}. For a long period of time, it has
been regarded as a problem that is often solved case by case relying
on experience. For example, \cite{rose2016research} suggests that
equity distribution should consider \textquotedblleft past and future
contributions,\textquotedblright{} but those contributions are very
subjective. To avoid this subjectivity, \cite{spolsky2011} argues
that everyone who joins the startup at the same time should receive
equal shares.

Recent years there are some theories and practices trying to deal
with this problem in a systematic way. To date, the most suitable
theory is the Shapley value in coalitional game theory, in which payoffs
are distributed according to the contribution of each of the sub-coalitions
and the three axioms of fairness \cite{shapley1988value}. An online
tool, \textquotedblleft Startup Equity Calculator,\textquotedblright{}
\cite{ha_2010} implements this idea by asking the question \textquotedblleft What
will the company look like without this particular founder?\textquotedblright ,
which essentially evaluates the contribution of each of the sub-coalitions.

However, the above ideas of the Shapley value assume that everyone
will agree on the contribution of each of the sub-coalitions. In practice,
different people have different opinions about the contribution of
each of the sub-coalitions, hence a coalitional game theory with incomplete
information is required, such as the one in this paper.

As another example, in the United States, passing legislation requires
substantial effort and extensive lobbying and debates, and the same
game is not played repeatedly. Thus, a repeated game model in the
paper \cite{chalkiadakis2004bayesian} is not applicable. In contemporary
U.S. politics, in addition, it is usually the case that the Bayesian
core, defined in \cite{ieong2008bayesian} does not exist, because
the two parties have strong prejudices about each other. Moreover,
the opinion exchange process affects the outcomes substantially, so
a model, such as an opinion consensus model, is required to capture
its effect. At the end, each Senator and Representative has her or
his own interests and cares almost exclusively about the welfare of
his or her own constituents. 

\section{Conclusions and Future work}

In this paper, a new framework for coalitional games is presented
with an unrealized subset payoff function and information exchange
among players. The framework creates an interplay between the traditional
model of the coalitional game and the opinion consensus model. Many
interesting implications arise from the new framework, including the
sufficient condition of non-stable core and the sufficient condition
of efficient information exchange. Furthermore, the case of algorithmic
learning players was studied, and the results were compared and connected
to the case of pure rational players. 

In the future, the dependency of equilibrium on the topology of the
opinion consensus network may be considered. It is clear that different
communication topologies will result in different steady states. From
the perspective of an investor in the business scenario, there is
a need to design a communication topology and rule (mechanism) that
ensures truth telling. From the perspective of the participants, questions
may arise concerning with whom and in what order should issues be
addressed to ensure favorable outcomes. Additional future work that
is needed is related to algorithmic learning. In this approach, quantizing
the rewards associated with the possible \textquotedblleft exit\textquotedblright{}
action of each player also could be considered. 

\bibliographystyle{IEEEtran}
\bibliography{reference}

\end{document}